\documentclass[letterpaper, 10 pt, conference]{ieeeconf}  % Comment this line out if you need a4paper

\IEEEoverridecommandlockouts                              % This command is only needed if 
\usepackage{graphics} % for pdf, bitmapped graphics files
\usepackage{epsfig} % for postscript graphics files
\usepackage{newtxtext,newtxmath}
\usepackage{times} % assumes new font selection scheme installed
\usepackage{amsmath} % assumes amsmath package installed

\usepackage{amssymb}  % assumes amsmath package installed
\usepackage{array}
\usepackage{stfloats}

\newcommand\numberthis{\addtocounter{equation}{1}\tag{\theequation}}

\newtheorem{definition}{Definition}{}
\newtheorem{proposition}{Proposition}{}
\newtheorem{condition}{Condition}{}
\newtheorem{remark}{Remark}{}
\newenvironment{proof}{{\noindent\it\bf Proof}\quad}{\hfill $\blacksquare$\par}

\usepackage{algorithm}
\usepackage{algorithmic}
\usepackage{subfig}
\usepackage{booktabs}

\usepackage[flushleft]{threeparttable}
\usepackage[justification=centering,labelsep=period]{caption}
\usepackage{bm}
\usepackage{enumerate}
\usepackage[colorlinks,
linkcolor=blue,
anchorcolor=blue,
citecolor=blue,
]{hyperref}%%修改此处为你想要的颜色

\title{\LARGE \bf
Numerically Stable Dynamic Bicycle Model for Discrete-time Control
}

\author{Qiang Ge$^{1}$, Shengbo Eben Li$^{1}$, Qi Sun$^{1}$ and Sifa Zheng$^{1*}$ % <-this % stops a space
\thanks{*This work is supported by National Key R\&D Program of China with 2018YFB1600600. All correspondence should be sent to S. Zheng.}% <-this % stops a space
\thanks{$^{1}$Qiang Ge, Shengbo Eben Li, Qi Sun and Sifa Zheng are with the School of Vehicle and Mobility, Tsinghua University, Beijing, China. (gq17@mails.tsinghua.edu.cn, \{lishbo, zsf\}@tsinghua.edu.cn, qisun@mail.tsinghua.edu.cn)}
}

\begin{document}

\maketitle
\thispagestyle{empty}
\pagestyle{empty}

%%%%%%%%%%%%%%%%%%%%%%%%%%%%%%%%%%%%%%%%%%%%%%%%%%%%%%%%%%%%%%%%%%%%%%%%%%%%%%%%
\begin{abstract}

Dynamic/kinematic model is of great significance in decision and control of intelligent vehicles. However, due to the singularity of dynamic models at low speed, kinematic models have been the only choice under many driving scenarios. This paper presents a discrete dynamic bicycle model feasible at any low speed utilizing the concept of backward Euler method. We further give a sufficient condition, based on which the numerical stability is proved. Simulation verifies that (1) the proposed model is numerically stable while the forward-Euler discretized dynamic model diverges; (2) the model reduces forecast error by up to 49\% compared to the kinematic model. As far as we know, it is the first time that a dynamic bicycle model is qualified for urban driving scenarios involving stop-and-go tasks.

%Dynamic/Kinematic model is of great significance in decision and control of intelligent vehicles. For a simplified single-track vehicle, dynamic model is more accurate taking account of tire sideslip, whereas kinematic model has an absolute advantage in low-speed stability. This paper manages to derive a discretized dynamic bicycle model feasible at low speed utilizing the concept of backward Euler method. We further give a sufficient condition, based on which the numerical stability is proved. Simulation verifies that (1) the proposed model is numerically stable when forward-Euler discretized dynamic model diverges; (2) the model reduces forecast error up to 49\% compared to kinematic model; (3) the model has comparable computational efficiency to that of kinematic model. As far as we know, it is the first time that a dynamic bicycle model is proposed, which is qualified for urban driving scenarios involving stop-and-go tasks.
%While various tire force model can be easily applied, the linear sideslip force is computationally efficient thanks to its explicit form, which is verified by simulation from two aspects: numerical stability and numerical accuracy.

\end{abstract}

%%%%%%%%%%%%%%%%%%%%%%%%%%%%%%%%%%%%%%%%%%%%%%%%%%%%%%%%%%%%%%%%%%%%%%%%%%%%%%%%
\section{Introduction}

Model, especially one meets these requirements: simplicity, differentiability, nonsingularity, numerical stability, is widely used in sensing, decision, and control of automated vehicles. For instance, such model describes how the state evolves without noise in an unscented  Kalman filter \cite{antonov2011unscented}, simulates kinematically-feasible trajectories of ego vehicle\cite{polack2017kinematic} or surrounding vehicles \cite{cui2019deep}, and forecasts future states in a Model Predictive Controller (MPC) \cite{falcone2007predictive}. Moreover, model-based reinforcement learning, such as Approximate Dynamic Programming (ADP), also relies on a state transition model \cite{lin2020continuous}.

The modeling of a typical front-steered four-wheel ground vehicle has been developed at various scales. Generally, they can be divided into two categories, i.e., kinematic and dynamic model \cite{rajamani2011vehicle}. Both contain a vast set of models, varying from 1-Degree-of-Freedom (DoF) integrator \cite{zheng2015stability} to 15-DoF full car \cite{perrelli2020analysis}. This paper focuses on the longitudinal and lateral control of vehicle, which involves 3 DoFs in the $x-y$ planar coordinates. Thus, the simplest model is the 3-DoF bicycle, also known as single-track model.

The two types of bicycle model are applied in different scenarios or tasks. Kinematic model works at low speed, especially stop-and-go scenarios, despite its simplification on tire sideslip characteristics. It is believed to be irreplaceable because all dynamic model encounter a singularity point when approaching zero speed. The tire slip angle estimation term has the vehicle velocity in the denominator \cite{kong2015kinematic}. Actually, even at 8 m/s the dynamic model can be unstable (see Fig. \ref{back_vs_for} (c), forward Euler method). Therefore, though more accurate at certain speeds, dynamic bicycle model can only be adopted when speed is relatively high or the discretization step length is short enough.

An insight into practical control/planning problems requires us to transform the continuous differential kinematic/dynamic bicycle model into a discrete version, i.e. difference equations. Initial Value Problem (IVP) of Ordinary Differential Equations (ODE) is dealt with many mature mathematical techniques, like the Euler method, Runge-Kutta methods, and linear multistep methods \cite{butcher2008numerical}. With forward Euler method $X_{k+1}$ can be expressed as an explicit linear combination of the derivative $f(X_k)$, the current state $X(k)$ and the step length $T_s$, which accommodates any form of $f(X_k)$ well. Though less accurate than other explicit methods, it is popular for single-step recurrence. Backward Euler method, one of the implicit methods, usually has better numerical stability but poor computational efficiency, because the next step state has to be approximated by iteration. More importantly, in control-oriented applications like MPC or ADP, the model must have an explicit form of forward propagation, limiting the backward Euler method to offline simulation.

As a result, a kinematic model discretized by forward Euler method is almost the default approach for autonomous parking \cite{xu2017model}, autonomous intersection management \cite{li2018near} and multi-vehicle formation \cite{cai2019multi}. However, the modeling error will inevitably grow with driving speed, as tire sideslip angle is no longer negligible.

The rest of this paper is organized as follows. In Section \ref{sec_continuous}, both kinematic and dynamic single-track models are illustrated. Section \ref{sec_discrete} derives an explicit discrete dynamic bicycle model with due simplifications, which is inspired by the idea of backward Euler method. Further, we manage to give a sufficient condition of provable numerical stability in section \ref{sec_stability}. The accuracy, stability, computational efficiency as well as support to stop-and-go tasks, are verified in Section \ref{sec_simulation}. Section \ref{sec_conclusion} gives some concluding remarks.

\section{Bicycle\ Model}\label{sec_continuous}

\subsection{Dynamic\ Model}

The continuous model is given in (\ref{dyna_continuous}). Here the longitudinal acceleration $a = \frac{T_w}{mR_w}$, where $T_w$ is the longitudinal drive torque from ground, $m$ is mass of the vehcle, $R_w$ is the rolling radius of the drive wheel. 

\begin{equation}
\dot{X}=f(X, U)=
\left[\begin{array}{c}
u \cos \varphi-v \sin \varphi \\
v \cos \varphi+u \sin \varphi \\
\omega \\
a+v \omega-\frac{1}{m} F_{Y 1} \sin \delta \\
-u \omega  + \frac{1}{m} \left (F_{Y 1} \cos \delta+F_{Y 2}\right)\\
\frac{1}{I_{z}}\left(l_f F_{Y 1} \cos \delta-l_r F_{Y 2}\right)
\end{array}\right]
\label{dyna_continuous}
\end{equation}

where 

\begin{equation}
    X=\left[\begin{array}{l}
    x\\
    y\\
    \varphi\\
    u\\
    v\\
    \omega
    \end{array}\right], \ U=\left[\begin{array}{l}
    a\\
    \delta
    \end{array}\right]
\end{equation}

% \begin{figure}[htbp]
%     \centerline{\includegraphics[width=5cm,height=6.25cm]{bicycle_new.pdf}}
%     \caption{Simplified bicycle model (dynamic)}
%     \label{dyna_bicycle}
% \end{figure}

\begin{figure}[htbp]
    \centering
    \captionsetup[subfigure]{justification=centering}
        \subfloat[\textit{a}]{\label{dyna}\includegraphics[width=0.23\textwidth]{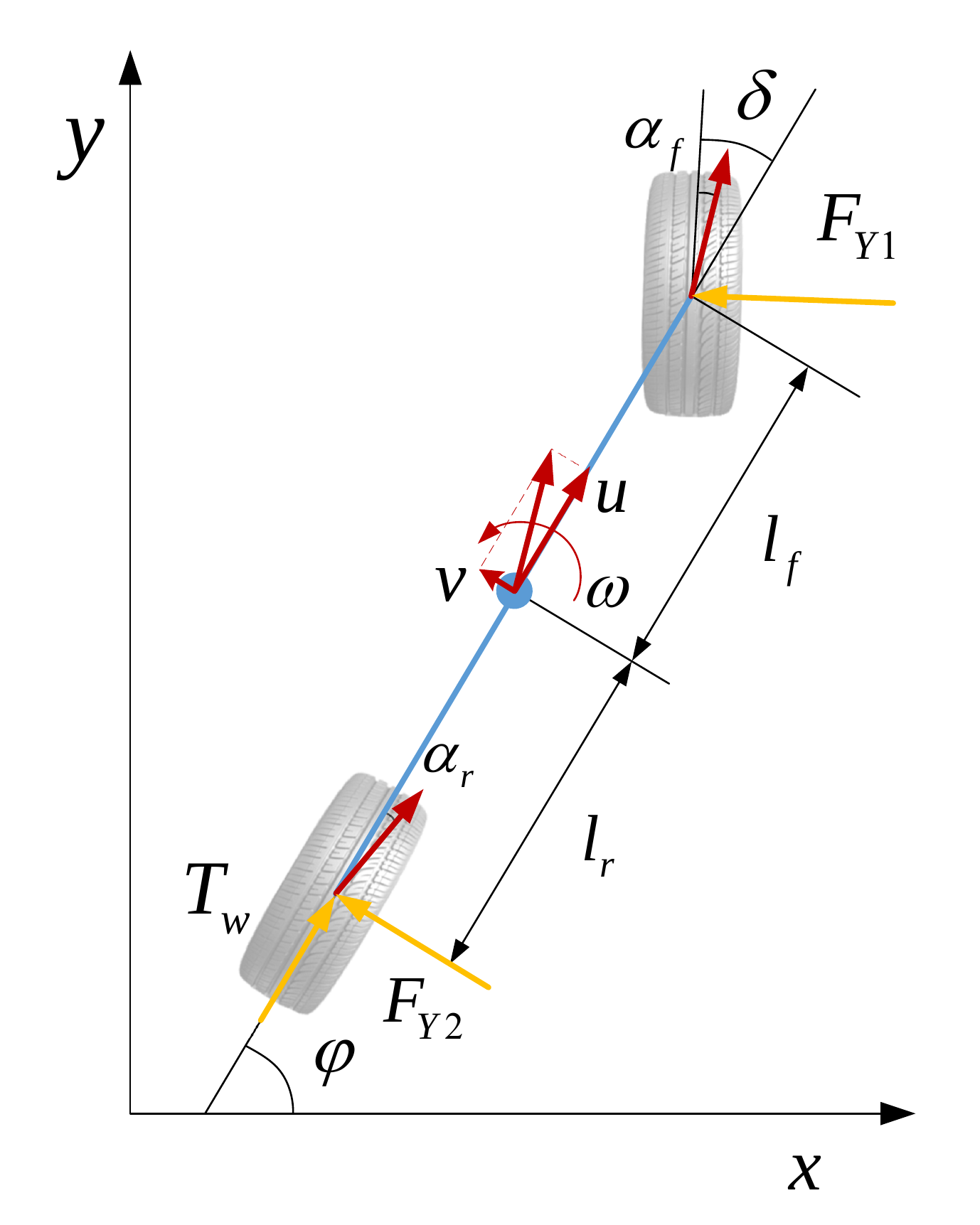}} \quad
        \subfloat[\textit{b}]{\label{kina}\includegraphics[width=0.23\textwidth]{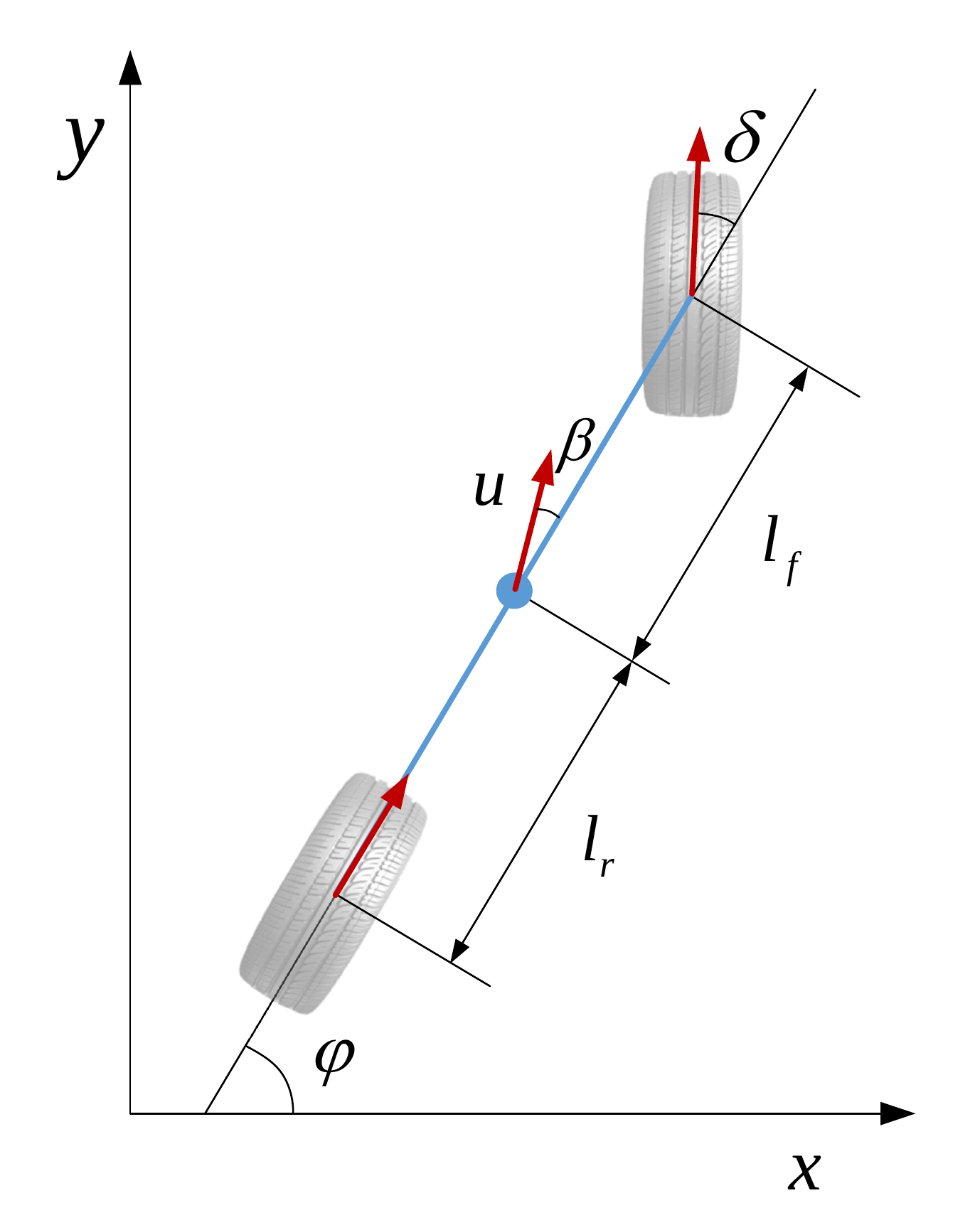}}
    %\captionsetup{font={footnotesize}}
    \caption{{ Bicycle model}\\
    (a) dynamic,\ (b) kinematic}
\label{bicycle models}
\end{figure}

The lateral forces $F_{Y1}$, $F_{Y2}$ are functions of sideslip angle $\alpha_{f}$ and $\alpha_{r}$. Under mild steering angle, we can assume that
\begin{subequations}
\begin{align}
    \cos\delta &\approx 1 \\
    \sin\delta &\approx 0\\
    F_{Y 1}&=k_{f} \alpha_{1} \approx k_{f}\left(\frac{v+l_{f} \omega}{u}-\delta\right)\label{c} \\
    F_{Y 2}&=k_{r} \alpha_{2} \approx k_{r} \frac{v-l_{r} \omega}{u}\label{d}
\end{align}
\end{subequations}

Note that we use a linear tire sideslip force model here, which is important for the deduction of an explicit discretized form later.

\subsection{Kinematic\ Model}

The widely used kinematic bicycle model \cite{rajamani2011vehicle} assumes that both front and rear wheels have only longitudinal rolling movement, but no lateral slip. This is inaccurate because tire characteristics are neglected.

\begin{equation}
    \dot{X}_{\text {kine }}=f_{\text{kine}}\left(X_{\text {kine }}, U\right)=\left[\begin{array}{c}
    u \cos (\varphi+\beta(\delta)) \\
    u \sin (\varphi+\beta(\delta)) \\
    a \\
    \frac{u}{l_{r}} \sin (\beta(\delta))
    \end{array}\right]
    \label{kine_continuous}
\end{equation}

where
\begin{equation}
    \beta(\delta)=\arctan \left(\tan \delta \frac{l_{r}}{l_{f}+l_{r}}\right)
\end{equation}

\begin{equation}
    X_{\text {kine}}=\left[\begin{array}{l}
    x\\
    y\\
    u\\
    \varphi\\
    \end{array}\right]
\end{equation}

Note that $u$ is longitudinal velocity in dynamic model, and absolute velocity in kinematic model.

\section{Discretization}\label{sec_discrete}

To apply the ordinary differential equations (\ref{dyna_continuous}), (\ref{kine_continuous}) to controller modeling or trajectory planning, the foremost step is to solve the initial value problems (IVP) with numerical methods. %While multistep methods like Runge-Kutta or Adams-Moulton methods are widely used for forward open-loop simulation, Euler method is the most popular solution for differentiable predictive models.

% \begin{equation}
%     \frac{X_{k+1}-X_{k}}{T_{s}}=f\left(X_{*}, U_{k}\right)
%     \label{euler}
% \end{equation}

%Forward Euler method is to substitute "$*$" with "$k$" in (\ref{euler}), whereas backward Euler method is to substitute with "$k+1$".
As a typical implicit method, the main defect of backward Euler method is that $X_{k+1}$ has to be found by solving the rootfinding problem (\ref{back_euler}), usually using fixed-point iteration.

\begin{equation}
    X_{k+1}-T_{s} f\left(X_{k+1}, U_{k}\right)=X_{k}
    \label{back_euler}
\end{equation}

In order to retain the advantage of explicit methods, we derive $X_{k+1}$ in a variable-by-variable fashion. $x_{k+1}$, $y_{k+1}$ and $\varphi_{k+1}$ are calculated by forward Euler method. $v_{k+1}$ and $\omega_{k+1}$ are solved from
\begin{small}
$$
\!v_{k+1} -T_s\!\left(-u_k \omega_k\!+\!\frac{1}{m} \left(F_{Y 1}(u_k,v_{k+1},\omega_k,\delta_{k})\!+\!F_{Y 2}(u_k,v_{k+1},\omega_k)\right)\right)\!=\!v_k
$$
$$
\omega_{k+1} -T_s \left( \frac{1}{I_{z}}\left(l_f F_{Y 1}(u_k,v_k,\omega_{k+1},\delta_{k}) -l_r F_{Y 2}(u_k,v_k,\omega_{k+1})\right)\right)=\omega_{k}
$$
\end{small}

Owing to the linear sideslip force assumption, the above formulas have analytical solutions. $u_{k+1}$ is special because theoretically, it belongs to a quadratic equation. Therefore, one item, i.e. $v \omega-\frac{1}{m} F_{Y 1} \sin \delta$ is simplified into zero. After that, $u_{k+1}$ is also derived by forward Euler method. Although such simplification introduces minor modeling error, in a predictive controller, the receding-horizon self-correction will help compensate for possible negative effects.
 
%For instance, $u_{k+1}$ is solved from $u_{k+1} = f_{u}(u_{k+1},v_{k},\omega_{k})$, rather than $u_{k+1} = f_{u}(u_{k+1},v_{k+1},\omega_{k+1})$.

Technically our discretization is not backward Euler method, but a variant inspired by it. It is the concept of backward Euler method that makes the model stable, which is proved in the next section. The nonlinear discrete model is

%Moreover, from (\ref{dyna_continuous}) to (\ref{back_euler_bicycle}), one item, i.e. $a+v \omega-\frac{1}{m} F_{Y 1} \sin \delta$ is simplified into $a_k$.
%The motivation is also to maintain an explicit expression of $u_{k+1}$.
%The motivation is also to avoid unfavorable iterations to solve $u_{k+1}$.

\begin{equation}
    X_{k+1}\triangleq F(X_{k},U_{k})
    \label{nonlinear_discrete}
\end{equation}

which is extended as:

\begin{equation}
    \left[\begin{array}{c}
    x_{k+1} \\
    y_{k+1} \\
    \varphi_{k+1} \\
    u_{k+1} \\
    v_{k+1} \\
    \omega_{k+1}
    \end{array}\right]=\left[\begin{array}{c}
    x_{k}+T_{s}\left(u_{k} \cos \varphi_{k}-v_{k} \sin \varphi_{k}\right) \\
    y_{k}+T_{s}\left(v_{k} \cos \varphi_{k}+u_{k} \sin \varphi_{k}\right) \\
    \varphi_{k}+T_{s} \omega_{k} \\
    u_{k}+T_{s} a_{k} \\
    \frac{m u_{k} v_{k}+T_{s}\left(l_{f}k_{f}-l_{r} k_{r}\right) \omega_{k}-T_{s} k_{f} \delta_{k} u_{k}-T_{s} m u_{k}^{2} \omega_{k}}{m u_{k}-T_{s}\left(k_{f}+k_{r}\right)} \\
    \frac{I_{z} u_{k} \omega_{k}+T_{s}\left(l_{f} k_{f}-l_{r} k_{r}\right) v_{k}-T_{s} l_{f} k_{f} \delta_{k} u_{k}}{I_{z} u_{k}-T_{s}\left(l_{f}^{2} k_{f}+l_{r}^{2} k_{r}\right)}
    \end{array}\right]
    \label{back_euler_bicycle}
\end{equation}

\section{Numerical\ Stability}\label{sec_stability}

\subsection{linearization\ of\ difference\ equation}

The nonlinear recursive difference equation (\ref{nonlinear_discrete}) can be linearized around a state point. Firstly, it can be seen as a reference point of the following nonlinear function
\begin{equation}
    Y = F(X,U)
\end{equation}

it can be derived that

\begin{equation}
    Y-X_{k+1}=\left.\frac{\partial F}{\partial X}\right|_{X=X_{k} \atop U=U_{k}}\left(X-X_{k}\right)+\left.\frac{\partial F}{\partial U}\right|_{X=X_{k} \atop U=U_{k}}\left(U-U_{k}\right)
    \label{linearization}
\end{equation}

denote $X-X_k$ as $\vec{\varepsilon_{k}}$, and $Y-X_{k+1}$ as $\vec{\varepsilon}_{k+1}$,

\begin{equation}
    \vec{\varepsilon}_{k+1}=J_{F, X}\left(X_{k}, U_{k}\right) \vec{\varepsilon}_{k} \triangleq A_{k} \vec{\varepsilon}_{k}
    \label{recursive_error}
\end{equation}

Note that from (\ref{linearization}) to (\ref{recursive_error}), the difference item of input $U$ was ignored. This is because the input $U$ is defined rather than estimated by us. In another word, we do not consider any error originating from $U-U_k$.

Additionally, note that the first three state variables, i.e. $x_k$, $y_k$ and $\varphi_k$ are merely combination and integral of the other variables $u_k$, $v_k$ and $\omega_k$. Thus, we can alternatively consider a dynamic system (still termed $F(X_k,U_k)$ for simplicity) with $X_k$ containing only $u_k$,$v_k$ and $\omega_k$, which distinguishes (\ref{back_euler_bicycle}) from other systems.

According to the definition of Jacobian matrix, under the reduced system,

%\begin{equation}
\begin{align*}
    A_k&= 
    \left[\begin{array}{ccc}
        \frac{\partial F}{\partial u_k} & \frac{\partial F}{\partial v_k} & \frac{\partial F}{\partial \omega_k} \end{array}\right] \\
    &=\left[\begin{array}{ccc}
    1 & 0 & 0 \\
    b_{1, k} & \frac{m u_{k}}{m u_{k}-T_{s}\left(k_{f}+k_{r}\right)} & \frac{T_{s}\left(l_f k_{f}-l_r k_{r}\right)-T_{s} m u_{k}^{2}}{m u_{k}-T_{s}\left(k_{f}+k_{r}\right)} \\
    b_{2, k} & \frac{T_{s}\left(l_f k_{f}-l_r k_{r}\right)}{I_{z} u_{k}-T_{s}\left(l_f^{2} k_{f}+l_r^{2} k_{r}\right)} & \frac{I_{z} u_{k}}{I_{z} u_{k}-T_{s}\left(l_f^{2} k_{f}+l_r^{2} k_{r}\right)}
    \end{array}\right]  \\
    &\triangleq\left[\begin{array}{cc}
        1 & \vec{0} \\
        \vec{b}_{k} & \hat{A}_k
    \end{array}\right] \numberthis
    \label{A_k}
\end{align*}
%\begin{equation}

\newcounter{TempEqCnt}                         % 创建临时变量TempEqCnt
\setcounter{TempEqCnt}{\value{equation}} % 将当前公式序号 赋给TempEqCnt
\setcounter{equation}{13}                           % 当前公式序号变为x，x等于长公式应有的序号减1.
\begin{figure*}[hb]
    \begin{equation}
        b_{1, k}=\frac{\left(m v_{k}-T_{s} k_{f} \delta_{k}-2 T_{s} m u_{k} \omega_{k}\right)\left[m u_{k}-T_{s}\left(k_{f}+k_{r}\right)\right]-m\left[m u_{k} v_{k}+T_{s}\left(l_f k_{f}-l_r k_{r}\right) \omega_{k}-T_{s} k_{f} \delta_{k} u_{k}-T_{s} m u_{k}^{2} \omega_{k}\right]}{\left[m u_{k}-T_{s}\left(k_{f}+k_{r}\right)\right]^{2}}
    \end{equation}
\end{figure*}

%\newcounter{TempEqCnt}                  % 创建临时变量TempEqCnt
\setcounter{TempEqCnt}{\value{equation}} % 将当前公式序号 赋给TempEqCnt
\setcounter{equation}{14}                           % 当前公式序号变为x，x等于长公式应有的序号减1.
\begin{figure*}[hb]
    \begin{equation}
        b_{2, k}=\frac{\left(I_{z} \omega_{k}-T_{s} l_f k_{f} \delta_{k}\right)\left[I_{z} u_{k}-T_{s}\left(l_f^{2} k_{f}+l_r^{2} k_{r}\right)\right]-I_{z}\left[I_{z} u_{k} \omega_{k}+T_{s}\left(l_f k_{f}-l_r k_{r}\right) v_{k}-T_{s} l_f k_{f} \delta_{k} u_{k}\right]}{\left[I_{z} u_{k}-T_{s}\left(l_f^{2} k_{f}+l_r^{2} k_{r}\right)\right]^{2}}
    \end{equation}
\end{figure*}

\subsection{analysis\ of\ numerical\ stability}

Our definition of numerical stability is given hereby,

\begin{definition}
    The system described by (\ref{back_euler_bicycle}) is called numerically stable if such boundedness characteristic is satisfied:
    $$ \forall \vec{\varepsilon}_{0} \in \mathcal{R}^{3 \times 1}, \exists C \in \mathcal{R}, s.t. \left\|\lim _{k \rightarrow \infty} A_{k} \cdots A_{1} A_{0} \vec{\varepsilon}_{0}\right\|\leq C $$
    \label{defi1}
\end{definition}

As a matter of fact, if (\ref{recursive_error}) can be seen as a switched discrete linear system, the concept of \textit{Joint Spectral Radius} could be of help in the stability proof \cite{jungers2009joint}. However, since the matrix $A_k$ in this case cannot be included by a finite set, we have to seek other methods.

Definition \ref{defi1} requires that, each entry of the long product $A_k\cdots A_1 A_0$ does not grow infinity with increasing $k$. This requirement, however, is difficult to be satisfied as the coefficient matrix $A_k$ is varying with each step $k$. More specifically, we believe that the rigorous stability only holds under certain conditions. Therefore, in what follows we manage to give a sufficient condition that will lead to this good feature.

As we have

\begin{equation}
    \begin{aligned}
    A_{k} \cdots A_{1} A_{0} &=\left[\begin{array}{cc}
    1 & \vec{0} \\
    \vec{b}_{k} & \hat{A}_{k}
    \end{array}\right] \cdots\left[\begin{array}{cc}
    1 & \vec{0} \\
    \vec{b}_{0} & \hat{A}_{0}
    \end{array}\right] \\
    &\triangleq\left[\begin{array}{cc}
    1 & \vec{0} \\
    \vec{b}_{k, 0} & \hat{A}_{k, 0}
    \end{array}\right]
    \end{aligned}
    \label{product}
\end{equation}

Thus, the numerical stability will be guaranteed as long as $\hat{A}_{k,0}$ and $\vec{b}_{k,0}$ are both bounded.

\subsubsection{boundedness\ of\ $\hat{A}_{k,0}$}

The sufficient condition is given as follows.

\begin{condition}\label{con1}
    $\left\|\hat{A}_{k}\right\| \leq 1, k=1,2, \cdots$
\end{condition}
%Condition 1: $\left\|\hat{A}_{k}\right\| \leq 1, k=1,2, \cdots$

If Condition \ref{con1} is satisfied, the long product $\hat{A}_{k,0}$ will decay exponentially. Moreover, we notice that $\hat A_k$ involves $u_k$ and step length $T_s$, which means that the norm of $\hat A_k$ will not be an arbitrary value but restricted in a certain range. This physical meaning could be crucial for the establishment of Condition \ref{con1}.

\begin{proposition}
    $\hat{A}_{k,0}$ is always bounded when $k$ approaches infinity, given that Condition \ref{con1} is satisfied.
\end{proposition}

\begin{proof}
\begin{equation}
    \begin{aligned}
    \left\|\hat{A}_{k, 0}\right\| &=\left\|\hat{A}_{k} \cdots \hat{A}_{1} \hat{A}_{0}\right\| \\
    & \leq\left\|\hat{A}_{k}\right\|\left\|\hat{A}_{1}\right\| \cdots\left\|\hat{A}_{0}\right\| \\
    & \leq 1
    \end{aligned}
\end{equation}
\end{proof}

\begin{remark}
    Given that Condition \ref{con1} is established, with the \textit{submultiplicative} property of induced norm, we prove that $\left\|\hat{A}_{k, 0}\right\|$ is always bounded by 1 as $k$ increases. Moreover, as long as $\left\|\hat{A}_{k}\right\| \not\equiv 1$, each entry of the submatrix $\hat{A}_{k,0}$ in (\ref{product}) will converge to zero exponentially, which is even more favorable for stability.
\end{remark}

\subsubsection{boundedness\ of\ $\vec{b}_{k,0}$}

\begin{proposition}
    $\left\|\vec{b}_{k, 0}\right\|$ is always bounded when $k$ approaches infinity, given that Condition \ref{con1} is satisfied.
\end{proposition}

\begin{proof}
\begin{equation}
    \begin{aligned}
    \vec{b}_{k, 0} &=\vec{b}_{k}+\hat{A}_{k} \vec{b}_{k-1,0} \\
    &=\vec{b}_{k}+\hat{A}_{k} \vec{b}_{k-1}+\hat{A}_{k} \hat{A}_{k-1} \vec{b}_{k-2,0} \\
    &=\cdots \\
    &=\vec{b}_{k}+\hat{A}_{k} \vec{b}_{k-1}+\cdots+\hat{A}_{k} \hat{A}_{k-1} \cdots \hat{A}_{1} \vec{b}_{0}
    \end{aligned}
\end{equation}

\begin{equation}\label{b_bound}
    \begin{aligned}
    \left\|\vec{b}_{k, 0}\right\| &=\left\|\vec{b}_{k}+\hat{A}_{k} \vec{b}_{k-1}+\cdots+\hat{A}_{k} \hat{A}_{k-1} \cdots \hat{A}_{1} \vec{b}_{0}\right\| \\
    & \leq\left\|\vec{b}_{k}\right\|+\left\|\hat{A}_{k} \vec{b}_{k-1}\right\|+\cdots+\left\|\hat{A}_{k} \hat{A}_{k-1} \cdots \hat{A}_{1} \vec{b}_{0}\right\| \\
    & \leq\left\|\vec{b}_{k}\right\|+\left\|\hat{A}_{k}\right\|\left\|\vec{b}_{k-1}\right\|+\cdots+\left\|\hat{A}_{k}\right\| \cdots\left\|\hat{A}_{1}\right\|\left\|\vec{b}_{0}\right\| \\
    & \leq\left\|\vec{b}_{*}\right\|+\left\|\hat{A}_{*}\right\|\left\|\vec{b}_{*}\right\|+\cdots+\left\|\hat{A}_{*}\right\|^{k}\left\|\vec{b}_{*}\right\| \\
    & =\left\|\vec{b}_{*}\right\| \frac{1-\left\|\hat{A}_{*}\right\|^{k+1}}{1-\left\|\hat{A}_{*}\right\|}
    \end{aligned}
\end{equation}

where

\begin{equation}\label{A_b_bound}
\begin{array}{l}
    \left\|\vec{b}_{*}\right\| = \max \limits_{u_{k}, v_{k}, \omega_{k}, \delta_{k}, T_s}\left\|\vec{b}_{k}\right\| \\
    \left\|\hat{A}_{*}\right\| = \max \limits_{u_{k},T_s}\left\|\hat{A}_{k}\right\|
\end{array}
\end{equation}

\end{proof}

Since the state variables $u_k$, $v_k$, $\omega_k$ are all bounded by their physical meanings, as continuous functions defined over them, $\left\|\vec{b}_{k}\right\|$ and $\left\|\hat{A}_{k}\right\|$ must have maximum values accordingly (\ref{A_b_bound}). According to our proof, given that Condition \ref{con1} is established, $\left\|\vec{b}_{k,0}\right\|$ is always bounded by $\frac{\left\|\vec{b}_{*}\right\|}{1-\left\|\hat{A}_{*}\right\|}$ as $k$ increases.
\begin{remark}
The proof is based on linearization of the initial nonlinear difference model. Nevertheless, according to \textit{Lagrange's mean value theorem}, the recursive formula (\ref{recursive_error}) can be changed into
\begin{equation}
    \begin{array}{c}
    \vec{\varepsilon}_{k+1}=A_{\operatorname{mid}} \vec{\varepsilon}_{k} \\
    A_{\operatorname{mid}}=J_{F, X}\left(X_{k}+(1-\tau) X_{k+1}, U_{k}\right), \tau \in[0,1]
    \end{array}
\end{equation}
And obviously, this will not change the subsequent analysis and conclusion.
\end{remark}

\begin{remark}
    Here we give a sufficient instead of sufficient \& necessary condition, which means that the violation of Condition \ref{con1} does not necessarily bring numerical instability. The main contribution of this condition, is the easy-to-satisfy property (see Fig. \ref{Ak_norm_vs_uk}) under general settings and common speed, which imposes great practical meaning to the application of dynamic bicycle model (\ref{back_euler_bicycle}).
\end{remark}

\section{Simulation}\label{sec_simulation}

As pointed out by Francesco Borrelli \textit{et al.} \cite{kong2015kinematic}, a kinematic model has these main advantages: low-speed feasibility and computational efficiency (2015). At the same time, they also observe significant growth of error on the higher speed sinusoidal and winding track tests. Arnaud de La Fortelle \textit{et al.} \cite{polack2017kinematic} then find that $a_y \leq 0.5\mu g$ is decisive, otherwise error of the kinematic bicycle model will become very large (2017).

Therefore, it is necessary to verify the low-speed feasibility of the proposed model, and compare the differences between both models. %This paper do not specify the state boundary where the model is reliable, but focus on common urban driving speed. %In what follows, the rolling resistance of the tire and air drag force are dismissed to better reveal the essential difference between the two models.
In what follows, we choose a vehicle (C-Class, Hatchback 2017) in \textit{CarSim} as the prototype. Important parameters are listed in Table \ref{tab: paras}. Open-loop simulation is about step steer response, whereas the closed-loop simulation involves a nonlinear MPC controller undertaking a stop-and-go task. 

\begin{table}

    \caption{Simulation Parameters}
    \label{tab: paras}
    \begin{center}
    \begin{tabular}{p{1cm}<{\centering} p{4.3cm} p{1.7cm}}
    \hline
    \textbf{Parameter}&\textbf{Description}&\textbf{Value}\\
    \hline

    %{$g$} & acceleration of gravity & 9.81 m/s$^2$\\

    {$I_{z}$} & yaw inertia of vehicle body & 1536.7 kg$\cdot$m$^2$\\

    {$k_f$} & front axle equivalent sideslip stiffness & -128916 N/rad\\

    {$k_r$} & rear axle equivalent sideslip stiffness & -85944 N/rad\\

    {$l_f$} & distance between C.G. and front axle & 1.06 m\\

    {$l_r$} & distance between C.G. and rear axle & 1.85 m\\

    {$m$} & mass of the vehicle & 1412 kg\\

    {$N_p$} & predictive horizon& 20\\

    {$N_c$} & control horizon& 1\\

    {$T_s$} & discretization step length& 0.1 s\\

    %{$T_c$} & control interval & 0.1 s\\

    %{$T_d$} & discretization interval & 0.1 s\\

    %{$\mu$} & friction coefficient & 0.85 \\

    \hline
    %\hline
    \end{tabular}
    \end{center}
\end{table}

\subsection{Open-loop\ Control}

\subsubsection{backward\ versus\ forward}

The numerical stability, as well as the insensitivity regarding step length and driving speed, is rather important for autonomous driving application. Before comparison, we first verify if the numerical stability Condition \ref{con1} is satisfied. Let 0 m/s $\leq\ u_k\ \leq$ 15 m/s, the 2-norm of $\hat{A}_k$ is shown in Fig. \ref{Ak_norm_vs_uk}. Since $\left\|\hat{A}_k\right\|_2\ \leq 1$, the numerical stability of our model is essentially guaranteed. 
\begin{figure}[htbp]
    \centerline{\includegraphics[width=6cm,height=4.69cm]{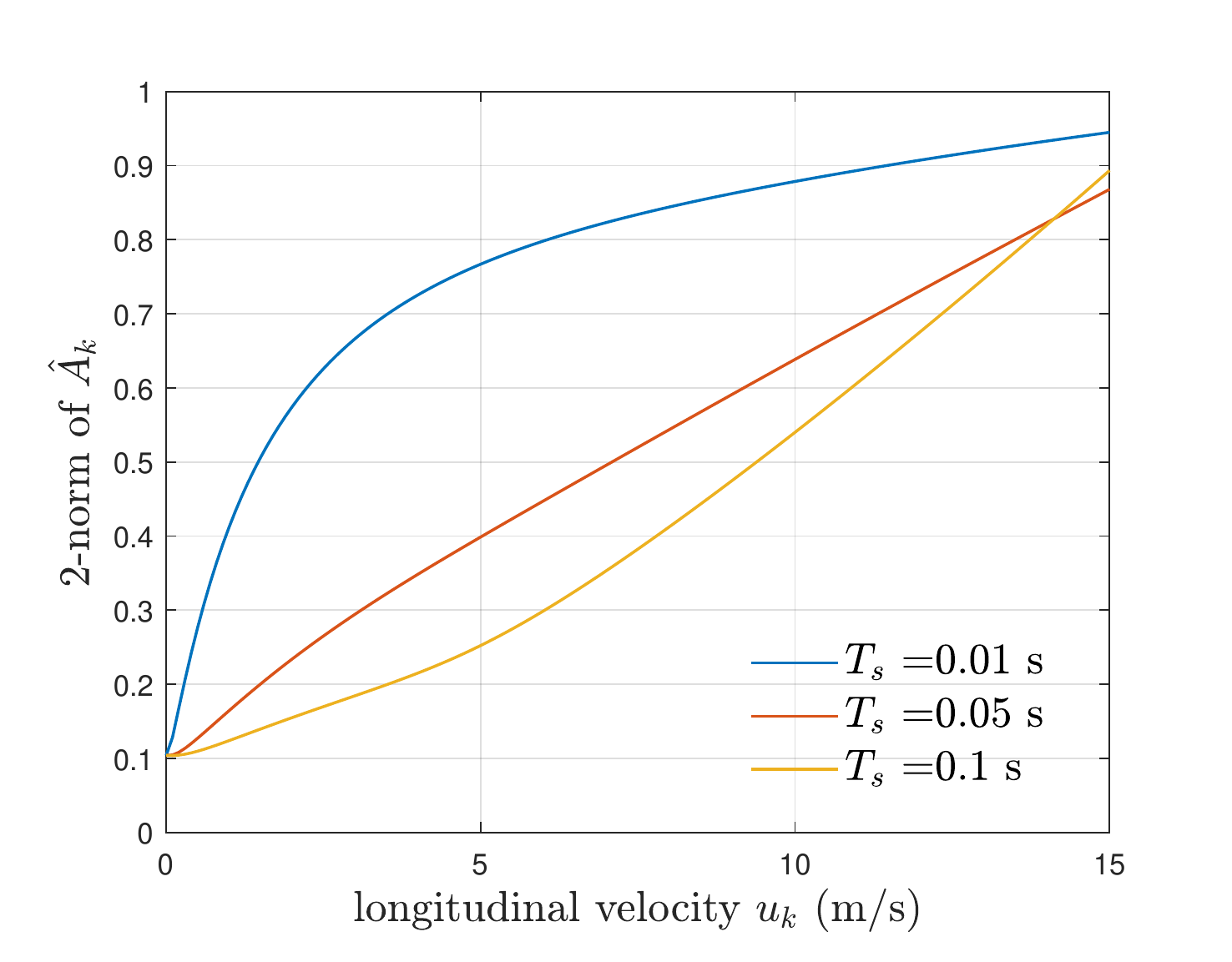}}
    \caption{Range of $\left\|\hat{A}_k\right\|_2$ defined over varying $u_k$}
    \label{Ak_norm_vs_uk}
\end{figure}

%The trade-off between computation time and real-time control determines $T_s$, therefore, the model should adapt to a range of $T_s$ and $u$ as much as possible.

\begin{figure}[htbp]
      \centering
      \captionsetup[subfigure]{justification=centering}
          \subfloat[\textit{a}]{\label{ts0.01}\includegraphics[width=0.4\textwidth]{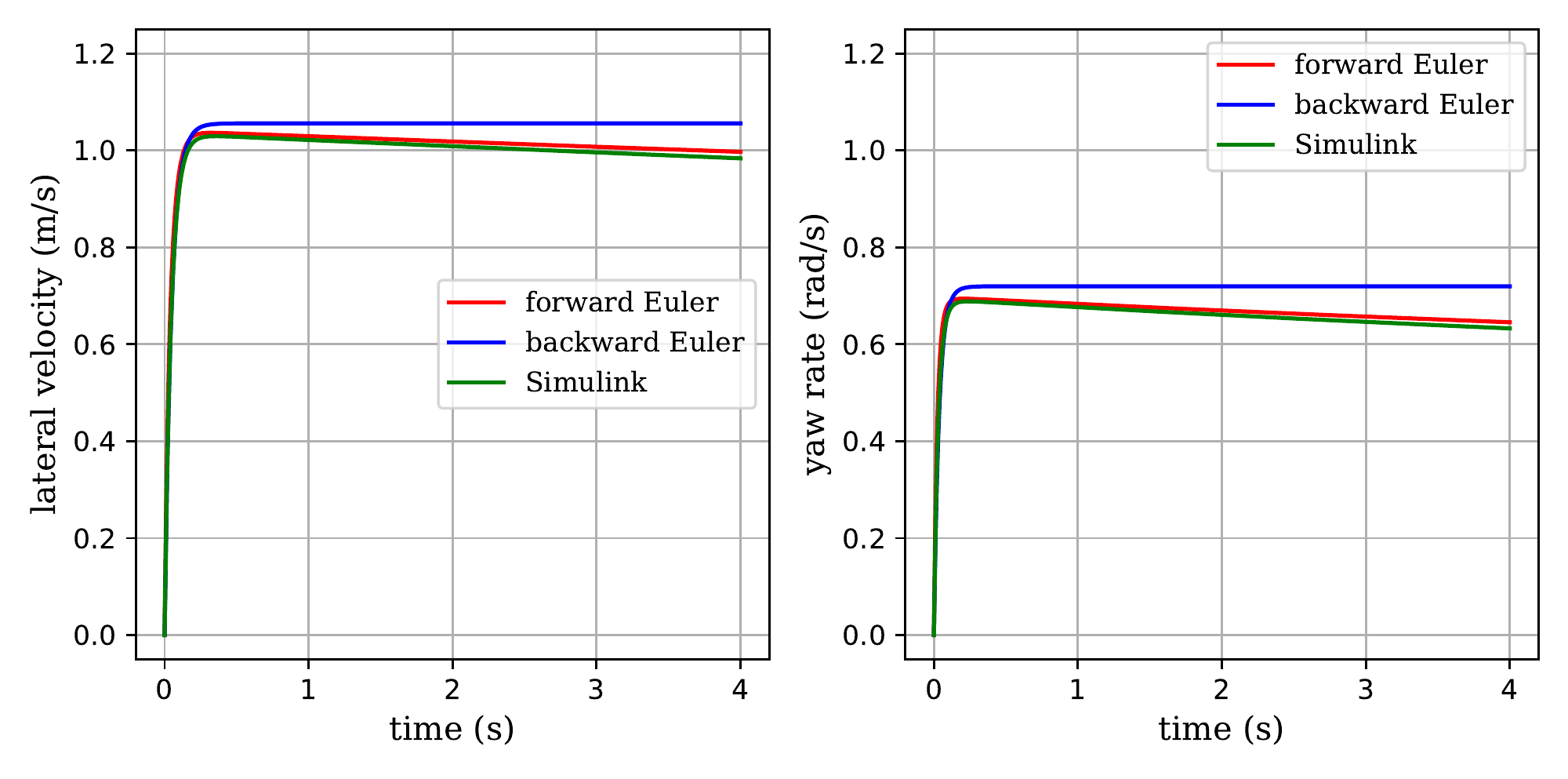}}\\
          \subfloat[\textit{b}]{\label{ts0.05}\includegraphics[width=0.4\textwidth]{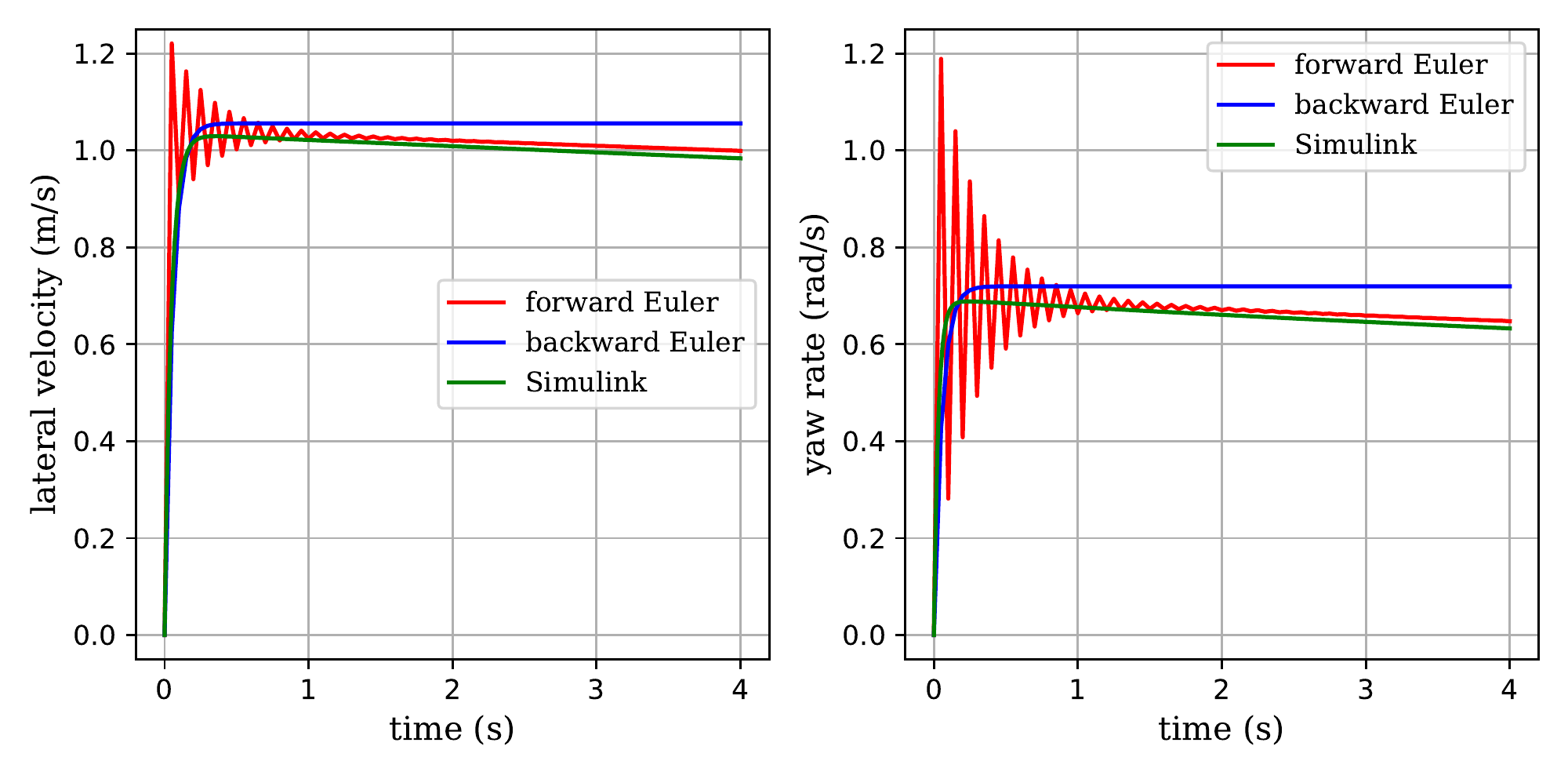}}\\
          \subfloat[\textit{c}]{\label{ts0.1}\includegraphics[width=0.4\textwidth]{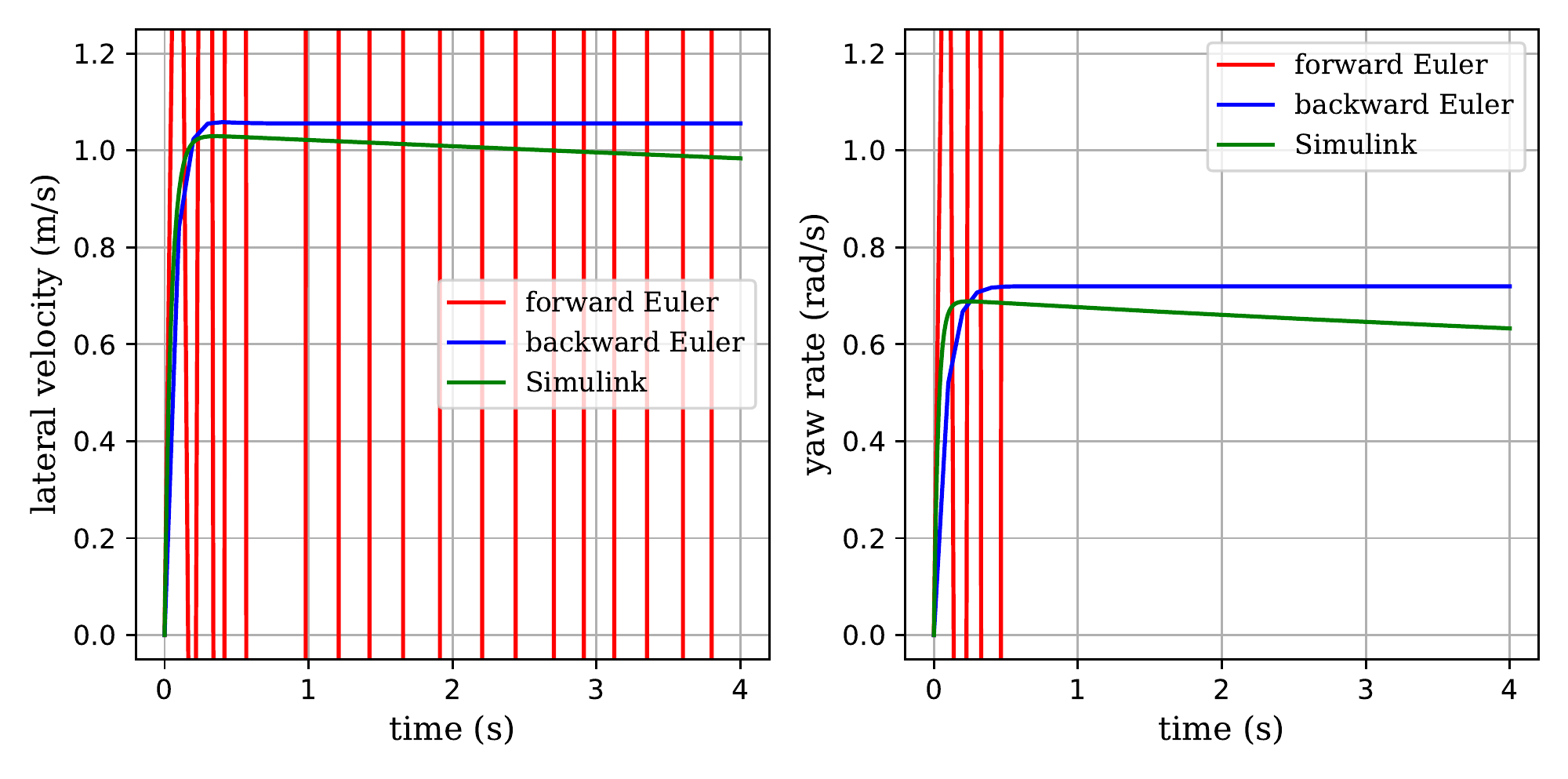}}\\
      %\captionsetup{font={footnotesize}}
      \caption{{Lateral state with varying discretization step length under step steering input. Conditions: $u_0$ = 8 m/s, $\delta$ = 0.2674 rad\\
      (a)\ $T_s$ = 0.01 s,\ (b) $T_s$ = 0.05 s,\ (c) $T_s$ = 0.1 s}}%\\
  \label{back_vs_for}
\end{figure}

The model (\ref{dyna_continuous}) with assumptions (\ref{c} and \ref{d}) is built in \textit{MATLAB/Simulink} with 4-order Runge-Kutta solver with a step length of 0.001s. The proposed model shows expected stability over time, and consistency over varying discretization step length $T_s$ (Fig. \ref{back_vs_for}). In contrast, the forward Euler method cannot hold this feature as $T_s$ increases.

\subsubsection{dynamic\ versus\ kinematic}

The open-loop step response of dynamic model (discretized by backward Euler method) and kinematic model (discretized by forward Euler method) are calculated, at different velocities. Simultaneously, the same input is imposed on a prototype model in \textit{CarSim} to approximate the groundtruth.

\begin{table}[htbp]

    \caption{Location RMS Error of 4-second Open-Loop Trajectory under Step Steer Input \\
    %Conditions: $T_s =$0.1s, $\delta$ = 0.2674 rad\\
    Prototype: \textit{CarSim} C-Class, Hatchback 2017}
    \label{tab: dyna_vs_kine}
    \begin{center}
    \begin{tabular}{m{1cm}<{\centering} m{1.8cm}<{\centering} m{1.8cm}<{\centering} m{1.8cm}<{\centering}} 
    \hline
    $u_0$ (m/s)&\textbf{backward dynamic (m)}&\textbf{forward kinematic (m)}&\textbf{Accuracy Improvement}\\
    \hline

    1 & 0.31 & \textbf{0.28} & -$11\%$\\

    2 & 0.42 & \textbf{0.38} & -$11\%$\\

    3 & 0.37 & \textbf{0.36} & -$3\%$\\

    4 & \textbf{0.60} & 0.73 & \textcolor{red}{+$18\%$}\\

    5 & \textbf{0.72} & 1.12 & \textcolor{red}{+$36\%$}\\

    6 & \textbf{0.93} & 1.72 & \textcolor{red}{+$46\%$}\\

    7 & \textbf{1.29} & 2.55 & \textcolor{red}{+$49\%$}\\

    8 & \textbf{1.83} & 3.62 & \textcolor{red}{+$49\%$}\\

    9 & \textbf{2.62} & 4.98 & \textcolor{red}{+$47\%$}\\

    10 & \textbf{3.89} & 6.85 & \textcolor{red}{+$43\%$}\\

    \hline
    %\hline
    \end{tabular}
    \end{center}
\end{table}

The open-loop trajectories of both models and their error relative to the prototype in \textit{CarSim} are shown in Fig. \ref{open_loop}.

\begin{figure}[htbp]
    \centering
    \captionsetup[subfigure]{justification=centering}
        \subfloat[\textit{a}]{\label{traj}\includegraphics[width=0.23\textwidth]{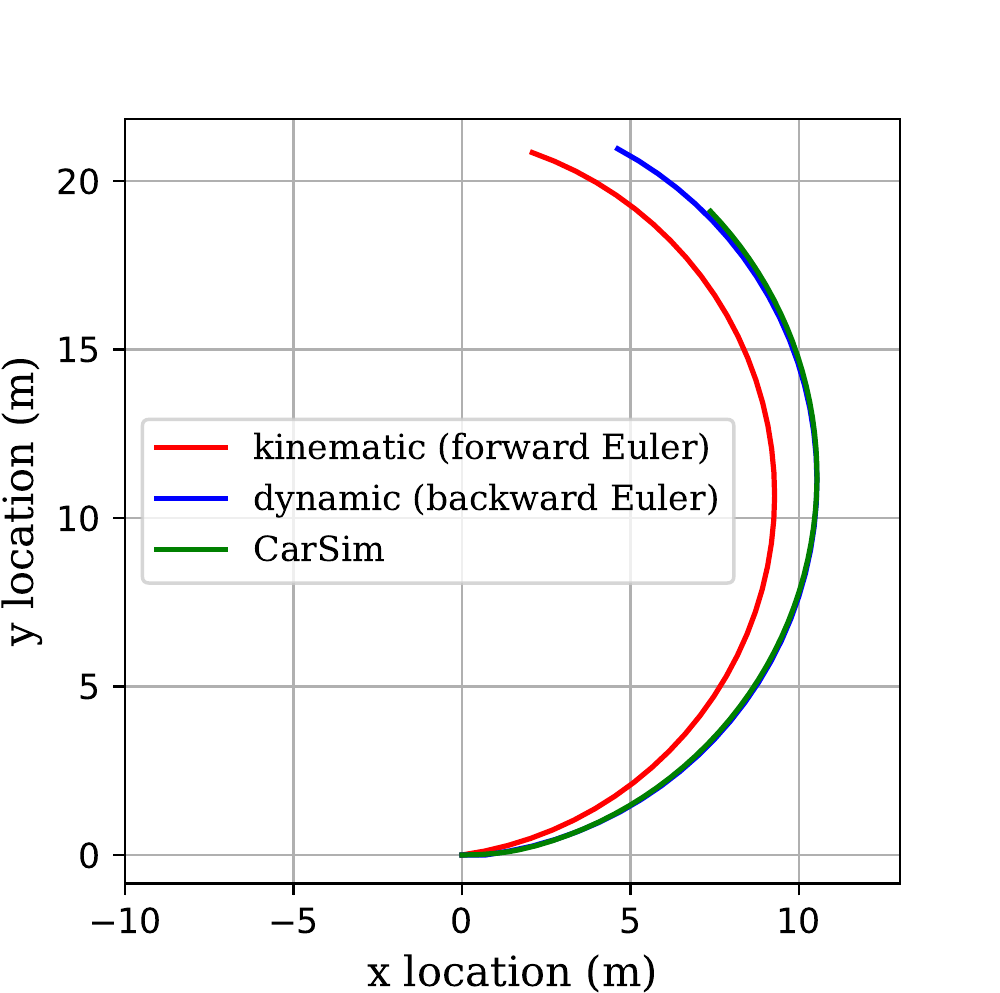}} \quad
        \subfloat[\textit{b}]{\label{error}\includegraphics[width=0.23\textwidth]{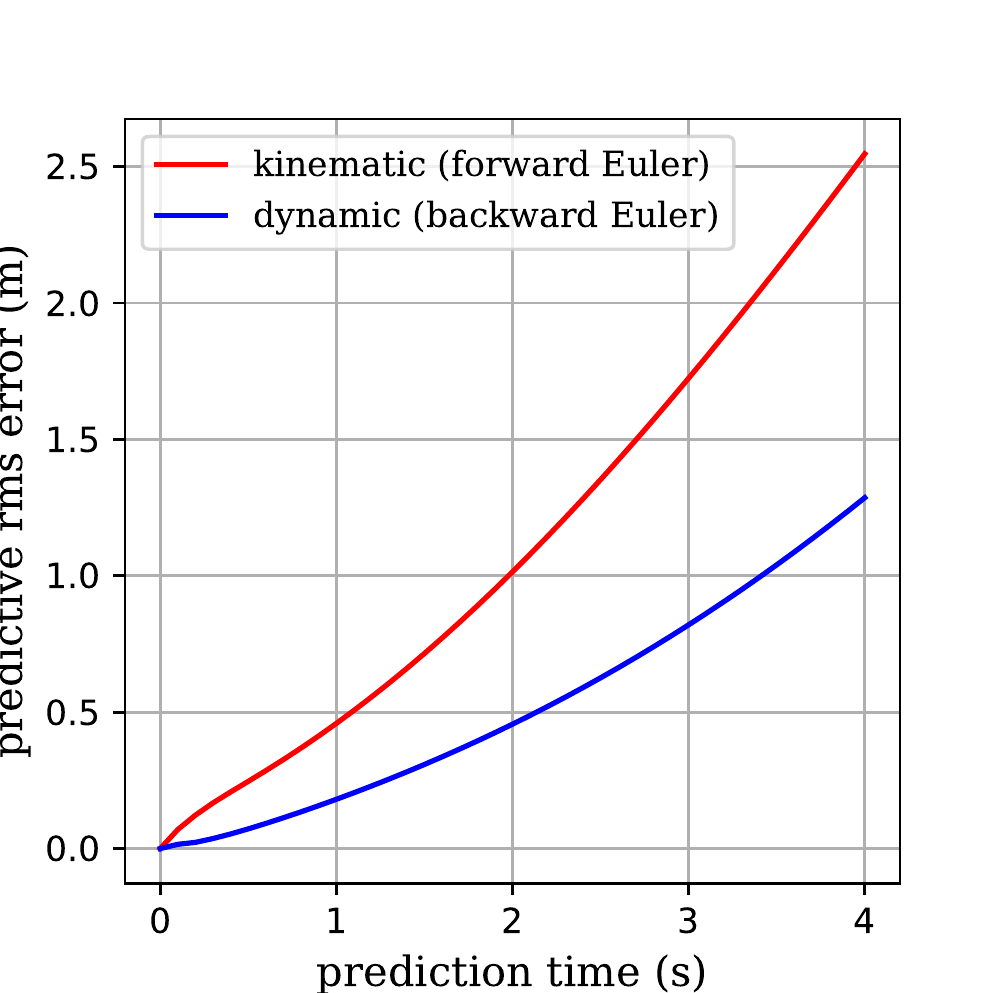}}
    %\captionsetup{font={footnotesize}}
    \caption{Open-loop trajectory and error under step steer input. Conditions: $u_0$ = 8 m/s, $\delta$ = 0.2674 rad\\
    (a) trajectory,\ (b) location error}
\label{open_loop}
\end{figure}

Although kinematic model outperforms the backward-discretized dynamic model at lower speeds, this situation is inversed since $u_0=4$ m/s under defined parameters. To be emphasized, the absolute error at higher speed is much bigger than that of lower speeds. More importantly, the prediction accuracy is crucial at higher speeds where occupants’ lives are at stake.%The dynamic model proposed is valuable, which can be explained from two aspects: it can be applied to stop-and-go scenarios; it is generally more accurate than the kinematic model. More importantly, it is obvious that the prediction accuracy is crucial at higher speeds where occupants’ lives are at stake.

\subsection{Closed-loop\ Control}

In this section, we design a nonlinear model predictive controller with (\ref{back_euler_bicycle}) to comprehensively verify its effect in practical application. Starting from (0,0), the vehicle model is always tracking a direct path connecting its center of gravity (C.G.) and the target (30,30). The reference speed is set to be 6 m/s. There is an obstacle (light blue round area in Fig. \ref{closed_loop_traj}) located at (15,15) blocking its way. The parameters are deliberately co-designed so that the vehicle can not bypass with steering, but only stop to avoid collision. Note that collision is defined by overlapping of the obstacle and C.G. of the vehicle, instead of vehicle body profile. Once the vehicle stops, the obstacle is moved to (18,12) (dark blue round area in Fig. \ref{closed_loop_traj}), giving way to the vehicle. Then the vehicle starts to execute a stop-and-go task. To be emphasized, the moved obstacle still requires it to steer and accelerate simultaneously, which is similar to a left-turning maneuver after red light at intersections.

\begin{figure}[htbp]
    \centerline{\includegraphics[width=6cm,height=5.52cm]{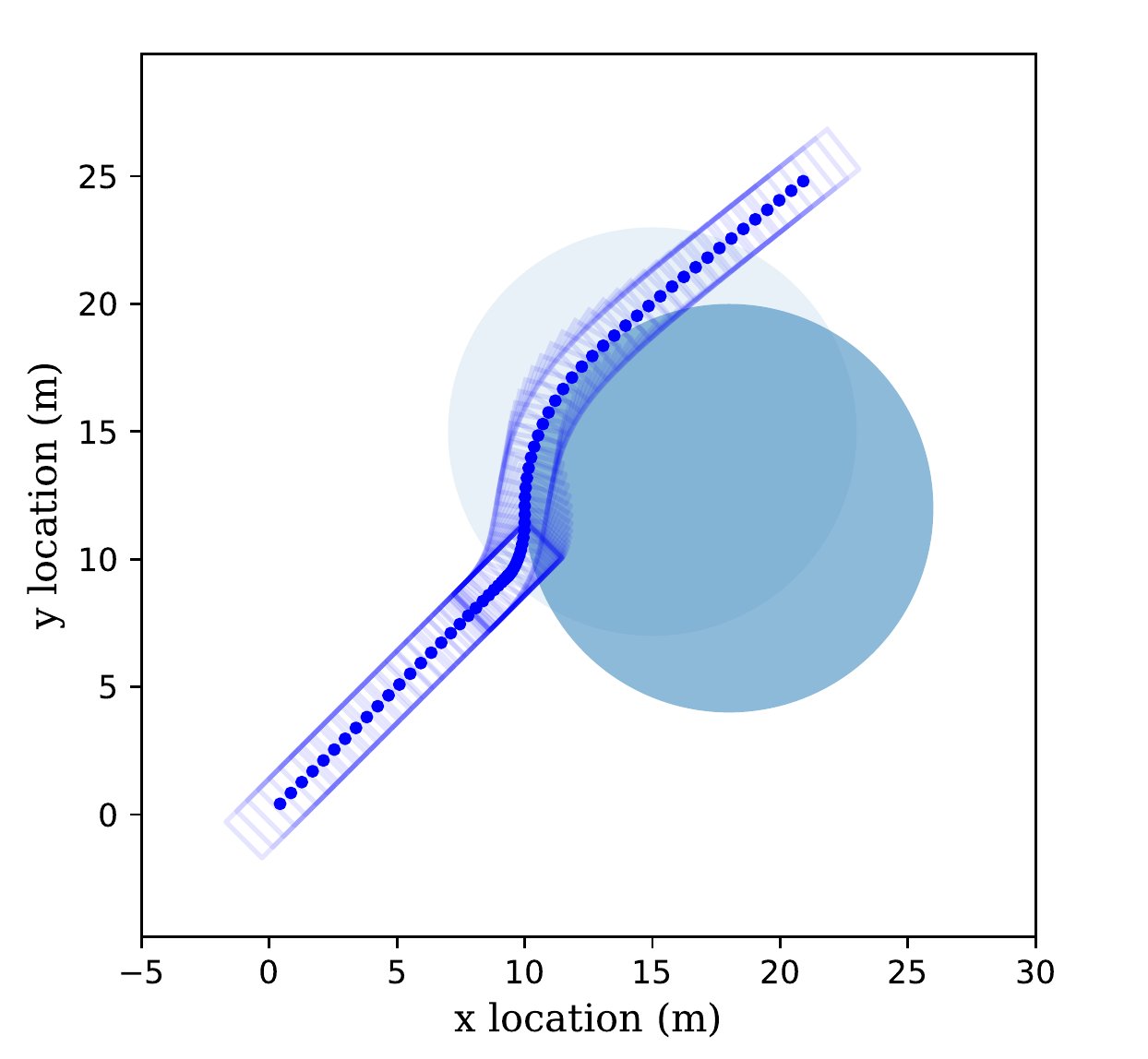}}
    \caption{Trajectory of the vehicle, avoiding a moving obstacle}
    \label{closed_loop_traj}
\end{figure}

The longitudinal velocity is shown in Fig. \ref{vel}, whereas the longitudinal input $a$ and lateral input $\delta$ are recorded in Fig. \ref{closed_loop}(a) and (b) respectively.

\begin{figure}[htbp]
    \centerline{\includegraphics[width=8cm,height=2.82cm]{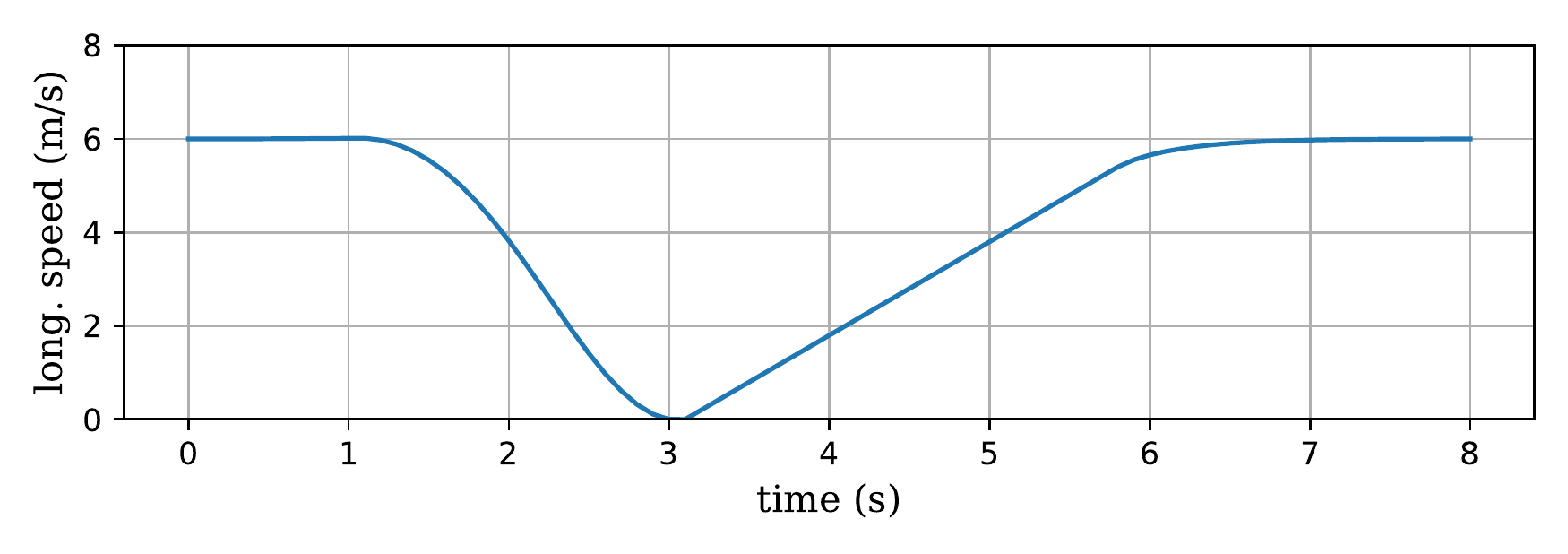}}
    \caption{Longitudinal velocity of the vehicle}
    \label{vel}
\end{figure}

\begin{figure}[htbp]
    \centering
    \captionsetup[subfigure]{justification=centering}
        \subfloat[\textit{a}]{\label{acc}\includegraphics[width=0.45\textwidth]{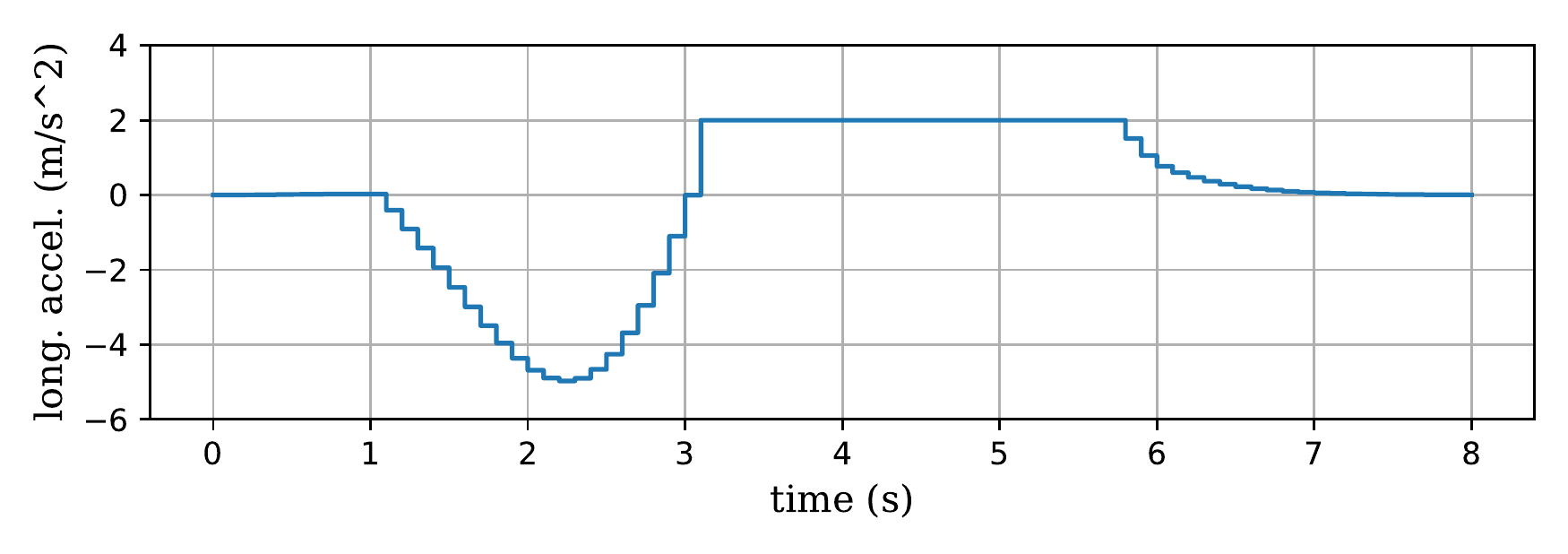}}\\
        \subfloat[\textit{b}]{\label{steer}\includegraphics[width=0.45\textwidth]{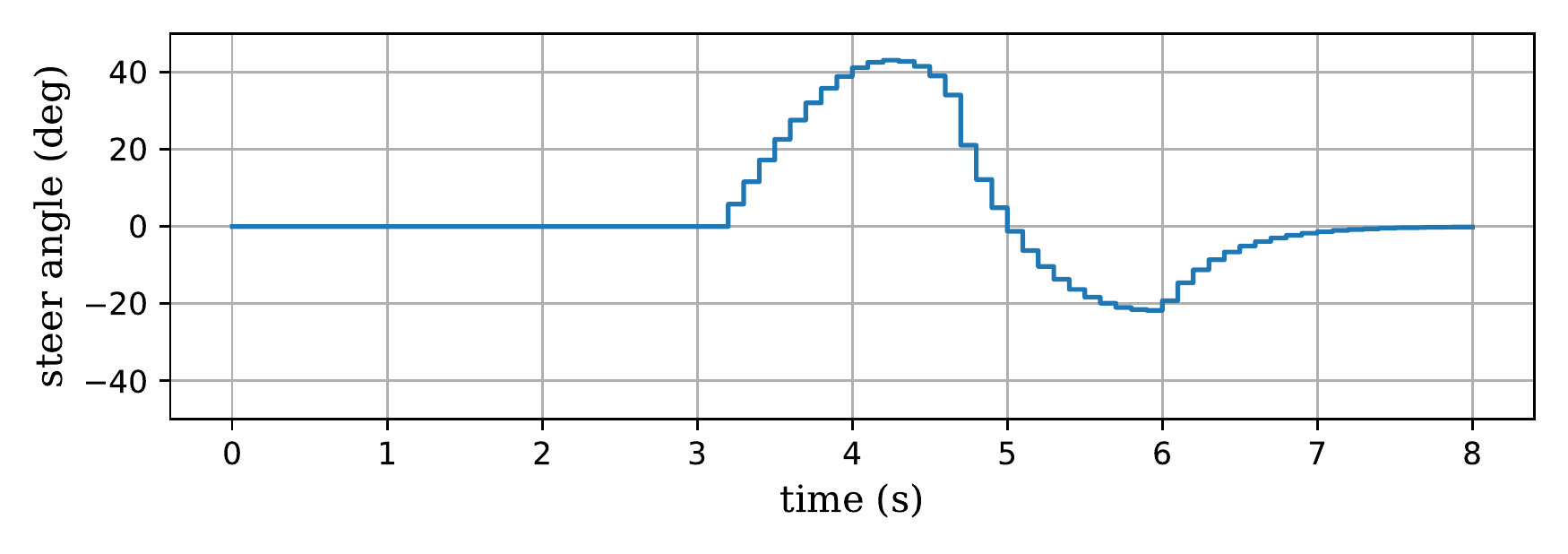}}\\
    %\captionsetup{font={footnotesize}}
    \caption{MPC control input\\
    (a) longitudianl input,\ (b) lateral input}
\label{closed_loop}
\end{figure}

The constrained nonlinear optimal control problem is formulated as (\ref{obj}), satisfying (\ref{csrt}).
\begin{equation}
    \min \sum_{k=0}^{N}\left(X(k)-X_{R}(k)\right)^{T} Q\left(X(k)-X_{R}(k)\right)+U(k)^{T} R U(k)
    \label{obj}
\end{equation}
\begin{subequations}
    \begin{align}
        X(k+1)&=F(X(k), U(k))\\
        \left(X(k)-X_{o b s}(k)\right)^{T} &Q_{s}\left(X(k)-X_{o b s}(k)\right) \geq D_{s}^{2} \\
        X_{\min } \leq &X(k) \leq X_{\max } \\
        U_{\min } \leq &U(k) \leq U_{\max }
    \end{align}
    \label{csrt}
\end{subequations}
$X_R(k)$ is the $k$th reference trajectory point in the predictive horizon, $X_{obs}(k)$ conveys location of the obstacle at $k$th step in the predictive horizon. $Q$ = diag(100,100,0,0,0,0), $R$ = diag(10,500), $Q_s$ = diag(1,1,0,0,0,0), $D_s$ = 8. $X_{min}$ = [-$\inf$,-$\inf$,-$\inf$,0,-4,-3], $X_{max}$ = [+$\inf$,+$\inf$,+$\inf$,20,4,3]. $U_{min}$ = [-5,-$\pi$/4], $U_{max}$ = [2,$\pi$/4]. All quantities are in the international system of units.

Our model is computationally efficient as well. The average computation time of one-step MPC solution is 61.6 ms, whereas the same simulation carried out with kinematic model (\ref{kine_continuous}) has an average of 59.6 ms. The problem is solved by the nonlinear OCP solver \textit{ipopt} \cite{wachter2006implementation} under the framework \textit{CasADi} \cite{andersson2019casadi}, on a laptop with an \textit{intel i7} CPU and 12 GB ram.

\section{Conclusion}\label{sec_conclusion}

%A conclusion section is not required. Although a conclusion may review the main points of the paper, do not replicate the abstract as the conclusion. A conclusion might elaborate on the importance of the work or suggest applications and extensions.

In this paper, we propose a discretized dynamic bicycle model inspired by backward Euler method and a sufficient condition that guarantees its numerical stability. The condition turns out easy to achieve under general vehicle parameters and common urban driving speed (e.g. $\leq 15 m/s$). Our method outperforms forward Euler method in stability and shows an advantage in accuracy up to 49\% compared to the kinematic model. A comprehensive stop-and-go task further verifies its effectiveness. It could be applied to vast fields, like MPC, ADP, etc., under general urban scenarios.

\addtolength{\textheight}{-12cm}   % This command serves to balance the column lengths

%%%%%%%%%%%%%%%%%%%%%%%%%%%%%%%%%%%%%%%%%%%%%%%%%%%%%%%%%%%%%%%%%%%%%%%%%%%%%%%%

%%%%%%%%%%%%%%%%%%%%%%%%%%%%%%%%%%%%%%%%%%%%%%%%%%%%%%%%%%%%%%%%%%%%%%%%%%%%%%%%

%%%%%%%%%%%%%%%%%%%%%%%%%%%%%%%%%%%%%%%%%%%%%%%%%%%%%%%%%%%%%%%%%%%%%%%%%%%%%%%%
%\section*{APPENDIX}


\section*{ACKNOWLEDGMENT}
%The preferred spelling of the word  acknowledgment  in America is without an  e  after the  g . Avoid the stilted expression,  One of us (R. B. G.) thanks . . .   Instead, try  R. B. G. thanks . Put sponsor acknowledgments in the unnumbered footnote on the first page.
The authors are grateful to Hao Chen and Xuewu Lin for their valuable advice in matrix product derivation.
%%%%%%%%%%%%%%%%%%%%%%%%%%%%%%%%%%%%%%%%%%%%%%%%%%%%%%%%%%%%%%%%%%%%%%%%%%%%%%%%

%References are important to the reader; therefore, each citation must be complete and correct. If at all possible, references should be commonly available publications.

% \bibliographystyle{IEEEtran}
% \bibliography{sample.bib}

\end{document}